\documentclass[a4paper,UKenglish,cleveref,autoref,thm-restate]{lipics-v2021}

\pdfoutput=1 
\hideLIPIcs 

\graphicspath{{./../fig/}}
\usepackage{graphicx}
\usepackage[T1]{fontenc}
\usepackage{amsmath}
\usepackage{enumerate}
\usepackage{ascmac}

\newcommand{\out}{\mathit{output}}
\newcommand{\occ}{\mathit{occ}}
\newcommand{\cons}{\mathit{cons}}
\newcommand{\ST}{\mathit{ST}}
\newcommand{\str}{\mathsf{str}}
\newcommand{\locus}{\mathsf{locus}}
\newcommand{\suffixtree}{\mathcal{T}}

\newcommand{\apex}{\mathit{apex}}
\newcommand{\rnk}{\mathit{rnk}}

\title{Data Structures for Range Sorted Consecutive Occurrence Queries}

\author{Waseem Akram}{Department of Computer Science and Engineering, Indian Institute of Technology, Kanpur, India}{akram@iitk.ac.in}{https://orcid.org/0009-0007-4960-8673}{}
\author{Takuya Mieno}{Department of Computer and Network Engineering, University of Electro-Communications, Chofu, Japan}{tmieno@uec.ac.jp}{https://orcid.org/0000-0003-2922-9434}{}
\authorrunning{Akram~and~Mieno} 

\Copyright{Waseem Akram and Takuya Mieno} 



\category{} 

\relatedversion{} 

\nolinenumbers

\EventEditors{John Q. Open and Joan R. Access}
\EventNoEds{2}
\EventLongTitle{42nd Conference on Very Important Topics (CVIT 2016)}
\EventShortTitle{CVIT 2016}
\EventAcronym{CVIT}
\EventYear{2016}
\EventDate{December 24--27, 2016}
\EventLocation{Little Whinging, United Kingdom}
\EventLogo{}
\SeriesVolume{42}
\ArticleNo{23}

\begin{document}

\maketitle

\begin{abstract}
  The \emph{string indexing problem} is a fundamental computational problem with numerous applications,
  including information retrieval and bioinformatics.
  It aims to efficiently solve the pattern matching problem: given a text $T$ of length $n$ for preprocessing
  and a pattern $P$ of length $m$ as a query,
  the goal is to report all occurrences of $P$ as substrings of $T$.
  Navarro and Thankachan~[CPM 2015, Theor. Comput. Sci. 2016] introduced
  a variant of this problem called the \emph{gap-bounded consecutive occurrence query},
  which reports pairs of consecutive occurrences of $P$ in $T$
  such that their gaps (i.e., the distances between them)
  lie within a query-specified range $[g_1, g_2]$.
  Recently, Bille et al.~[FSTTCS 2020, Theor. Comput. Sci. 2022] proposed
  the \emph{top-$k$ close consecutive occurrence query},
  which reports the $k$ closest consecutive occurrences of $P$ in $T$, sorted in
  non-decreasing	order of distance.
  Both problems are optimally solved in query time with $O(n \log n)$-space data structures.
  In this paper, we generalize these problems to the \emph{range query model},
  which focuses only on occurrences of $P$ in a specified \emph{substring} $T[a.. b]$ of $T$.
  Our contributions are as follows:
  \begin{itemize}
    \item We propose an $O(n \log^2 n)$-space data structure that answers the \emph{range} top-$k$ consecutive occurrence query in $O(m + \log\log n + k)$ time
      and can be built in $O(n\log^3 n)$ time;
      and
    \item We propose an $O(n \log^{2+\epsilon} n)$-space data structure that answers the \emph{range} gap-bounded consecutive occurrence query in $O(m + \log\log n + \out)$ time
      and can be constructed in $O(n\log^{5/2}n)$ time,
      where $\epsilon$ is a positive constant and $\out$ denotes the number of outputs.
  \end{itemize}
  As by-products, we present algorithms for geometric problems
  involving weighted horizontal segments in a 2D plane, which are of independent interest.
  Furthermore, we observe that consecutive occurrences are related to \emph{closed substrings} of a string.
\end{abstract}
\section{Introduction}\label{sec:intro}
The \emph{string indexing problem} involves constructing a data structure that efficiently supports \emph{pattern matching queries}.
Given a text $T$ of length $n$ and a pattern $P$ of length $m$,
the goal is to report all occurrences of $P$ as substrings of $T$.
This problem has numerous applications in the real world
including information retrieval and bioinformatics.
A typical example of such an index structure is the \emph{suffix tree}~\cite{Weiner73},
a compact trie of all suffixes of $T$.
When combined with perfect hashing~(e.g.,~\cite{FredmanKS82,Ruzic08}),
a suffix tree provides an $O(n)$-space data structure that answers pattern matching queries
in optimal $O(m + \occ)$ time for given pattern $P$,
where $\occ$ is the number of occurrences of $P$ in $T$.

In this work, we focus on variants of pattern matching queries
that consider the distances between the starting positions of every two occurrences of $P$ in $T$.
Keller et al.~\cite{KellerKL07} introduced a \emph{non-overlapping} pattern matching problem,
which reports only occurrences of $P$ separated by at least $m$ characters.
They proposed a data structure of size $O(n \log n)$ that reports the maximal sequence of non-overlapping occurrences of $P$ in $T$ sorted in textual order,
in $O(m + \out\cdot\log\log n)$ query time where $\out$ is the output size.
Cohen and Porat~\cite{CohenP09} later improved the query time to $O(m + \out)$.
Several subsequent studies extended their work~\cite{GangulyST20,HooshmandAKT21,GibneyMT23}.

In 2016, Navarro and Thankachan~\cite{NavarroT16} introduced another new variant of the pattern matching query
called the \emph{consecutive occurrence} query.
The query reports \emph{consecutive} occurrences of pattern $P$,
where each pair of occurrences has no other occurrences of $P$ between them
and the distance is within range $[g_1, g_2]$ specified by a query.
They designed an $O(n \log n)$-space data structure that supports the consecutive occurrence query in optimal $O(m + \out)$ time.
Subsequently, Bille et al.~\cite{Bille22} considered the \emph{top-$k$ close} consecutive occurrence query,
which reports (at most) $k$ consecutive occurrences of $P$ in 
non-decreasing order 
of distance, where $k$ is specified in a query.
They proposed an $O(n \log n)$-space data structure that answers the top-$k$ consecutive occurrence query in optimal $O(m + k)$ time. 
In the same work, they also presented two other trade-offs for the problem, each using $O(n)$ space, with query times of either $O(m+k^{1+\epsilon})$ or $O(m + k \log^{1+\epsilon} n)$.
Recently, Akram and Saxena~\cite{Akram24} revisited Bille et al.'s work
and proposed a \emph{simpler} data structure with the same time/space bounds, which is the basis of our proposed method in this paper.
Further, Bille et al.~\cite{Bille23} addressed an extension of the query that requires consecutive occurrences of two (possibly distinct) patterns $P_1$ and $P_2$
where their distances lie within range $[g_1, g_2]$.
They showed that the query can be answered in $\tilde{O}(|P_1|+|P_2|+n^{2/3}\occ^{1/3})$ time with $\tilde{O}(n)$ space,
and that any data structure using $\tilde{O}(n)$ space requires 
$\tilde{\Omega}(|P_1|+|P_2|+\sqrt{n})$ query time
unless the \emph{String SetDisjointness Conjecture}~\cite{GoldsteinKLP17} fails.
More recently, Gawrychowski, Gourdel, Starikovskaya and Steiner considered consecutive occurrence problems on grammar-compressed texts~\cite{GawrychowskiGSS23,Gawrychowski24}.
As we have mentioned, the \emph{consecutive} occurrence queries has attracted significant attention in recent years.

In this paper, we consider each variant of consecutive occurrence queries defined in~\cite{Bille22} and~\cite{NavarroT16}
within \emph{substrings} specified by queries.
The first one is the \emph{range top-$k$ consecutive occurrence} query
that reports the $k$ closest consecutive occurrences of pattern $P$ in \emph{substring} $T[a.. b]$, sorted by distance,
where the query input consists of pattern $P$, range $[a,b]$, and integer $k$.
The second one is the \emph{range gap-bounded consecutive occurrence} query
that reports all the consecutive occurrences of pattern $P$ in \emph{substring} $T[a.. b]$
where their distances lie within $[g_1, g_2]$,
where the query input consists of pattern $P$, range $[a, b]$ representing the substring of $T$, and gap range $[g_1, g_2]$.
The formal definitions of the problems are provided in Section~\ref{sec:pre}.
Our main results are as follows:
\begin{itemize}
  \item For the range top-$k$ consecutive occurrence query,
    we propose an $O(n\log^2 n)$-size data structure with a query time of $O(m + \log\log n + k)$~(Theorem~\ref{thm:topk}).
  \item For the range gap-bounded consecutive occurrence query,
    we propose an $O(n\log^{2+\epsilon} n)$-size data structure with a query time $O(m + \log\log n + \out)$,
    where $\epsilon$ is a positive constant~(Theorem~\ref{thm:gapbounded}).
\end{itemize}
Both of the data structures consist of the suffix tree~\cite{Weiner73} of $T$ 
and some appropriately chosen geometric data structures,
associated with heavy paths~\cite{Sleator83} of the suffix tree.
Although the high-level idea is similar to that of~\cite{NavarroT16},
we carefully employ non-trivial combinations of data structures
since our problems only require occurrences within the query substring and 
outputs needs to be sorted by distance unlike~\cite{NavarroT16}.
Further, as by-products, we develop algorithms for geometric problems
involving weighted horizontal segments in 2D plane, which are of independent interest~(Theorems~\ref{thm:geometrictopk} and \ref{thm:geometricgapbounded}).

Our data structures for the top-$k$ queries and the gap-bounded queries can be constructed in $O(n\log^3 n)$ time and $O(n \log^{5/2} n)$ time, respectively.

The main component of the construction algorithms is a simple method that efficiently computes sets of horizontal segments, which encode all possible consecutive occurrences in the string $T$ in a compact manner.
Additionally, we draw a connection between consecutive occurrences and the notion of closed substrings~\cite{Fici17} in Section~\ref{sec:closed_substrings}. We provide a compact representation for all closed substrings of a string $T$, whose size is $O(n \log n)$. Moreover, we also show the method for computing sets of horizontal segments can be adapted to compute the representation of closed substrings in $O(n \log n)$ time.    

Some of the results above appeared in a preliminary version of this paper~\cite{CPM2025paper}. As new contributions, the present article includes construction algorithms for the proposed data structures, as well as new insights into a relation between consecutive occurrences and closed substrings.

\subsection{Related Work}
\subsubsection*{Range Queries for Strings.}
Variants of range queries (a.k.a. internal or substring queries) have been proposed for various problems in string processing.
M\"akinen and Navarro~\cite{MakinenN07} proposed a data structure for a fundamental range query called the \emph{position-restricted substring searching query},
which reports the occurrences of pattern $P$ within a specified substring $T[a.. b]$.
Their solution requires $O(n\log^\epsilon n)$ space, and answers queries in $O(m + \log \log n + \occ)$ time.
Keller et al.~\cite{KellerKL07} introduced the range non-overlapping pattern matching query and
proposed a data structure of size $O(n\log n)$ that can answer queries in $O(m + k\log\log n)$ time.
Cohen and Porat~\cite{CohenP09} later improved these results to $O(n\log^\epsilon n)$ space and $O(m+\log\log n + k)$ query time.
Crochemore et al.~\cite{CrochemoreIKRTW12} proposed an alternative trade-off with $O(n^{1+\epsilon})$ space and $O(m+ k)$ query time.

Bille and G{\o}rtz~\cite{BilleG14} showed that several string problems, including the position-restricted substring searching query,
can be reduced to the \emph{substring range reporting} problem, and
provided an efficient solution for the problem.
One of the results from~\cite{BilleG14} is  
a data structure
of size $O(n\log^\epsilon n)$ that answers
the position-restricted substring searching query in optimal $O(m + \occ)$ time,
which is an improvement of  M\"akinen and Navarro's result~\cite{MakinenN07}.
Range queries have been considered not only for finding patterns in strings
but also for computing regularities in strings or compressing substrings
(see~\cite{Cormen22,Kociumaka16,CrochemoreIRRSW20,Abedin0PT20,AmirCPR20,BadkobehCKP22,MitaniMSH23,ShibataK24,KociumakaRRW24} and references therein).

\subsubsection*{Orthogonal Queries for Line Segments.}
Given a set $S$ of $N$ horizontal segments in the plane, the \emph{orthogonal segment intersection problem} involves designing an efficient data structure that can be used to find the segments intersected by a vertical query segment. Chazelle~\cite{Chazelle86FS} presented a solution that occupies $O(N)$ space and takes $O(\log N + K)$ time to answer a query, where $K$ is the output size. It is an important problem in computational geometry and has been studied in various settings (see~\cite{Rahul14,Shi05,Shi08}).

Rahul and Janardan~\cite{Rahul14} studied a variant of the orthogonal segments intersection problem where the input set $S$ consists of weighted horizontal segments and queries to be answered are of the following type: given a vertical segment and an integer $k$, report the $k$ segments intersected by the query segment with largest weights. They proposed a general technique to solve top-$k$ variants of geometric problems efficiently, and solve this problem using $O(N\log^2 N)$ space and $O(\log^3 N + k)$ query time. 

\subsection{Paper Organization.}
The rest of the paper is organized as follows.
Section~\ref{sec:pre} includes notations, definitions and results used in this paper. 
Section~\ref{sec:naive} describes na\"{i}ve solutions to both variants. 
Section~\ref{sec:topk} presents an efficient algorithm to the top-$k$ problem (Definition~\ref{def:topk}).
Section~\ref{sec:boundedgap} describes an efficient solution for the gap-bounded variant (Definition~\ref{def:gapbounded}). 
Section~\ref{sec:closed_substrings} discusses the relation between concepts of consecutive occurrence and closed substrings, and provides a compact representation of closed substrings. Finally, Section~\ref{sec:conclusions} presents some concluding remarks.

\section{Preliminaries}\label{sec:pre}
\subsection{Notations}
Let $\Sigma$ be an ordered alphabet.
An element in $\Sigma$ is called a character.
An element in $\Sigma^\star$ is called a string.
The length of a string $T$ is denoted by $|T|$.
The empty string $\varepsilon$ is the string of length $0$.
If $T = xyz$ holds for strings $T, x, y$, and $z$,
then $x, y$, and $z$ are called a prefix, a substring, and a suffix of $T$, respectively.
They are called a proper prefix, a proper substring, and a proper suffix of $T$
if $x \ne T$, $y \ne T$, and $z \ne T$, respectively.
A string $b$ that is both a proper prefix and a proper suffix of $T$ is called a border of $T$.
For a string $T$ and an integer $i$ with $1 \le i \le |T|$,
we denote by $T[i]$ the $i$th character of $T$.
Also, for integers $i, j$ with $1 \le i \le j \le |T|$,
we denote by $T[i.. j]$ the substring of $T$ that starts at position $i$ and ends at position $j$.
For convenience, we define $T[i.. j] = \varepsilon$ if $i > j$.
We sometimes use notation $T[i.. ] := T[i.. |T|]$ for suffixes.
For strings $T$ and $w$,
we denote by $\occ(w, T) = \{i\mid T[i.. i+|w|-1] = w\}$
the set of starting positions of occurrences of $w$ in $T$.
For two positions $i, j \in \occ(w, T)$ with $i < j$,
the pair $(i, j)$ is called a consecutive occurrence of $w$ in $T$
if there is no element $k \in \occ(w, T)$ with $i < k < j$.
For a consecutive occurrence $(i, j)$ we call the value $j-i$ the distance of $(i, j)$.
We denote by $\cons(w, T)$ the set of consecutive occurrences of $w$ in $T$.
Note that $|\cons(w, T)| = |\occ(w, T)| - 1$ if $\occ(w, T) \ne \emptyset$.

Now we formally define two problems we tackle in this paper.
\begin{definition}[Top-$k$ query]\label{def:topk}
  The range top-$k$ consecutive occurrence query over string $T$ is,
  given a query $(P, [a, b], k)$ where $P$ is a pattern and $a, b, k$ are integers,
  to output the top-$k$ close consecutive occurrences $(i, j) \in \cons(P, T[a.. b])$,
  sorted in 
  non-decreasing order of distance.
\end{definition}
\begin{definition}[Gap-bounded query]\label{def:gapbounded}
  The range gap-bounded consecutive occurrence query over string $T$ is,
  given a query $(P, [a, b], [g_1, g_2])$ where $P$ is a pattern and $a, b, g_1, g_2$ are integers,
  to output the consecutive occurrences $(i, j) \in \cons(P, T[a.. b])$
  with $g_1 \le j-i \le g_2$,
  sorted in 
  non-decreasing order of distance.
\end{definition}
We give a concrete example for the top-$k$ query in Figure~\ref{fig:ex}.
\begin{figure}[tbh]
  \centering
  \includegraphics[width=0.8\linewidth]{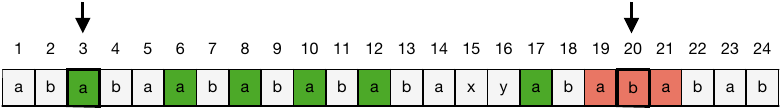}
  \caption{An example of the top-$k$ query for $k=4$. The occurrences of $P = \mathtt{aba}$ in the index range $[3, 20]$ are $3, 6, 8, 10, 12,$ and $17$. The occurrence of $P$ at index $19$ is invalid with respect to the range $[3, 20]$ as it ends at index $21 > 20$. Thus, the top-$4$ consecutive occurrences of $P$ in the index range $T[3.. 20]$ are $(6, 8), (8, 10), (10, 12)$, and $(3, 6)$.}
  \label{fig:ex}
\end{figure}

In what follows, we fix a string $T$ of length $n > 0$ arbitrarily.
In this paper, we assume the standard \emph{word-RAM model} with word size $w = \Omega(\log n)$.

\subsection{Algorithmic Tools}
\subsubsection*{Suffix Trees.}
The \emph{suffix tree} of string $T$ is a compacted trie for the set of suffixes of $T$~\cite{Weiner73}.
We denote by $\str(v)$ the \emph{path-string} obtained by concatenating the edge labels
from the root to node $v$ of the suffix tree of $T$.
If the last character of $T$ is unique (i.e., occurs in $T$ only once), 
then the suffix tree of $T$ has exactly $|T|$ leaves and
there is a one-to-one correspondence between the leaves and the suffixes of $T$.
More precisely, for each suffix $T[i..]$ of $T$,
there exists a leaf $\ell_i$ of the suffix tree of $T$
such that $\str(\ell_i) = T[i..]$, and vice versa.
Thus, we label such a leaf $\ell_i$ with integer $i$.
For a substring $w$ of $T$,
we define $\locus(w)$ as the highest node $u$ of the suffix tree of $T$
such that $w$ is a prefix of the string $\str(u)$.
The suffix tree of string $T$ over a general ordered alphabet can be constructed in $O(n\log n)$ time~\cite{Weiner73}.
Also, by using a perfect hashing,
one can compute in constant time the outgoing edge from any node
that starts with any specified character, after $O(n (\log \log n)^2)$ time preprocessing~\cite{Ruzic08}.

In the rest of this paper, we assume that the last character of $T$ is unique (one can achieve this property  by appending a new letter), and we denote by $\suffixtree$ the suffix tree of $T$. Namely, $\suffixtree$ has exactly $n$ leaves.

\subsubsection*{Heavy Path Decomposition of Tree.}
A {\em heavy path} of a tree is a root-to-leaf path in which each node has a size no smaller than the size of any of its siblings (the size of a node is the number of nodes in the subtree rooted at that node).
{\em Heavy path decomposition} is a process of decomposing a tree into heavy paths~\cite{Sleator83}. First, we find a heavy path by starting from the tree's root and choosing a child with the maximum size at each level. We follow the same procedure recursively for each subtree rooted at a node that is not on the heavy path, but its parent node is (on the heavy path). As a result, a collection of (disjoint) heavy paths is obtained. Some of its properties are:
\begin{itemize}
  \item each node $v$ belongs to exactly one heavy path,
  \item any path from the root to a leaf can pass through at most $\log_2 N$ heavy paths where $N$ is the size of the number of nodes in the tree,
  \item the number of heavy paths equals the number of leaves in the tree.
\end{itemize}
Heavy path decomposition of a rooted tree with $N$ nodes can be computed in $O(N)$ time~\cite{Sleator83}.
\subsubsection*{Segment Tree.}
The \emph{segment tree}~\cite{deBerg08LSI} is an important data structure in computational geometry. It stores a set of $N$ horizontal line segments in the plane such that the segments intersected by a vertical query line can be computed in $O(\log N + K)$ time, where $K$ is the number of reported segments. It can be built in $O(N\log N)$ space and time.
It is a balanced binary search tree of height $O(\log N)$
built over the $x$-coordinates of the segments’ endpoints. Each node stores a vertical slab $H(v):=[a,b)\times \mathbb{R}$ and a set $S(v)$ of segments, where $a$ and $b$ are real numbers. Let $l_1, l_2,\ldots l_m$ be the list of leaf nodes in the left-to-right order. The slab of a leaf node $l_i$
, for each $i=1,2,\ldots m-1$, 
is determined by the key values of $l_i$ and $l_{i+1}$, and the vertical slab of an internal node is the union of those of its children. Note that the vertical slabs of the nodes of any particular level form a partition of the plane. A segment $s$ is stored at each highest node $v$ such that $s \cap H(v)\ne \emptyset$ and no endpoint of $s$ lies inside $H(v)$. Given a real value $x$, the path from the root to the leaf $l_i$ whose vertical slab contains $x$ is referred to as the \emph{search path} for $x$. 
An important property of the segment tree is that a segment intersects the vertical line through $x$ if and only if it is stored at some node of the search path for $x$.

\subsubsection*{Fractional Cascading Technique.}
Chazelle and Guibas \cite{Chazelle86FC} introduced the \emph{fractional cascading technique}, a data structure technique which is useful in the scenarios where one needs to search an item in multiple sorted lists. It speeds up the search process.
We briefly describe the idea of the technique in the context where lists are associated with the nodes of a balanced binary tree, the reader may refer to \cite{Chazelle86FC} for full details.
Let $T$ be a balanced binary tree with $O(\log n)$ height, and each node $v$ stores a sorted list $L_v$ of comparable items
where $n$ is the combined size of all lists.
Given a root-to-leaf path in $T$, one can trivially find the position (rank) of an item in the list of each node on the path in $O(\log^2 n)$ time; performing a binary search at each node of the path. By employing fractional cascading technique, the search time can be reduced to $O(\log n)$. The technique introduces 
an \emph{augmented list} $A_v$
for each node $v$ in the tree: the augmented list $A_v$ consists of the items of $L_v$ and every fourth  
item from the list of each child; the list $A_v$ is also sorted. 
These augmented lists are built in a bottom-up fashion.
Each item $x$ in $A_v$ stores a constant amount of information:~(1) whether $x$ is from $L_v$ or from one of $v$’s children, (2) pointers to immediate neighbors of $x$ on either side in $A_v$ that belong to $L_v$, (3) a pointer to the appropriate item in $A_w$ for each child $w$ of $v$. If the position of $x$ in $A_v$ is known, one can find its position in the list of either child (of node $v$) in $O(1)$ time using the saved pointers.
Suppose we are to find the position (rank) of an item $\alpha$ in lists of nodes on a particular
root-to-leaf path $\pi = \langle u_0 = \mathit{root}, u_1, u_2, \ldots\rangle$.
A binary search is used to find the location of $\alpha$ in $A_{u_0}$
(and hence in $L_{u_0}$).
Then, using the pointers stored in $A_{u_0}$, the correct rank of $\alpha$ in the list $L_{u_1}$ of the next node $u_1$ on the path $\pi$ can be computed in $O(1)$ time.
Generally speaking, if the rank of $\alpha$ in $L_{u_j}$ is known, the rank of $i$ in $L_{u_{j+1}}$ can be obtained in $O(1)$ time using the pointers stored in $A_{u_j}$.
Thus, the total time spent would be $O(\log n)$.

\subsubsection*{Sorted Orthogonal Range Reporting Queries.}
A \emph{sorted range reporting} problem over an array $A$ of integers is,
given a query $([x_1, x_2], k)$ where $x_1, x_2, k$ are integers,
to report the $k$ smallest values in subarray $A[x_1.. x_2]$ in sorted order.
The following result is due to Brodal et al.~\cite{BrodalFGL09}. 
\begin{lemma}[\cite{BrodalFGL09}]\label{lem:brodal}
  There is a data structure of size $O(|A|)$
  that can answer any sorted range reporting query in $O(k)$ time.
  The data structure can be constructed in $O(|A|\log|A|)$ time.
\end{lemma}
A \emph{2D orthogonal range reporting} problem asks to represent a set of 2D planar points into an efficient data structure
such that all the points inside an orthogonal query rectangle can be reported efficiently.
Gao, He and Nekrich~\cite{Gao20} studied a variant of the problem where the output points are to be reported in increasing (or decreasing) order of one of the two coordinates.
They proposed an almost linear space solution to the problem that takes $O(\log \log N + K)$ time to answer a query, where $N$ is the size of the input set and $K$ is the number of reported points.
The output points are reported one by one in the sorted order. Their main result is as follows.

\begin{lemma}[\cite{Gao20}]\label{lem:sortedrangereporting}
  Given a set of $N$ points in the 2D rank space,
  one can preprocess in $O(N\sqrt{\log N})$ time the set so that the points lying inside a query rectangle $[a,b]\times [c,d]$ can be reported in sorted order
  in $O(\log \log N + K)$ time, where $K$ is output size. The space used by the data structure is $O(N\log^{\epsilon} N)$ words.
\end{lemma}

\subsubsection*{Dynamic Fusion Nodes.}
A \emph{dynamic fusion node}~\cite{PatrascuT14} is a data structure for a dynamic set of integers
which supports operations on the set including insertion, deletion, predecessor, and successor.
Every query can be executed in $O(\log N/\log w)$ time where $N$ is the size of the dynamic set and $w$ is the machine word size. 
The data structure uses $O(N)$ space. 
If $N$ is close to $w$, the following holds:
\begin{lemma}[\cite{PatrascuT14}]\label{lem:fusionnode}
  If $N = w^{O(1)}$, we can maintain a dynamic set of $N$ integers
  supporting insertion, deletion, predecessor, and successor in $O(1)$ time using $O(N)$ space.
\end{lemma}
We will use this lemma later in the case where $N \le \log |T| \in O(w)$.

\section{\texorpdfstring{Na\"{i}ve}{Naive} Solutions}\label{sec:naive}
In this section, we present a data structure that answers top-$k$ queries in near-optimal time. The data structure also supports gap-bounded queries in the same time bound. However, the space complexity is $O(n^2\log^{\epsilon} n)$.

We first build the suffix tree $\mathcal{T}$ over the text $T$ of length $n$.
Recall that the labels of leaves in the subtree rooted at a node $u$ correspond to the occurrences of $\str(u)$.
Let $C(u)$ denote the set of consecutive occurrences of $\str(u)$. For each $(i,j) \in C(u)$, we define a $2$-dimensional point $(i, j-i)$. Let $\mathsf{P}(u)$ be the set of all such $2$-d points at node $u$.
We then preprocess the set $\mathsf{P}(u)$, for each $u\in \mathcal{T}$, into a data structure that supports sorted range reporting queries using Lemma~\ref{lem:sortedrangereporting}.
The resulting data structure supports both types of queries.

Firstly, we estimate the space complexity.
The size of the set $C(v)$ of consecutive occurrences is one less than the number of leaves in the subtree rooted at node $v$.
In the worst case (when $\mathcal{T}$ is left-skewed or right-skewed), the total size of all sets $C(v)$, $v\in \mathcal{T}$, is $O(n^2)$.
As the set $\mathsf{P}(v)$ contains one point for each consecutive occurrence in the set $C(v)$, $|\mathsf{P}(v)|=|C(v)|$ holds.
Consequently, $\sum_{v\in \mathcal{T}} |\mathsf{P}(v)| = \sum_{v\in \mathcal{T}} |C(v)| = O(n^2)$.
The suffix tree $\mathcal{T}$ occupies $O(n)$ space.
The sorted range reporting structure representing the set $\mathsf{P}(v)$ uses $O(|\mathsf{P}(v)|\log^{\epsilon}|\mathsf{P}(v)|)$ space due to Lemma~\ref{lem:sortedrangereporting}.
Thus, the total space used by the data structures associated with all nodes is $\sum_{v\in \mathcal{T}} O(|\mathsf{P}(v)|\log^{\epsilon}|\mathsf{P}(v)|) = O(n^2 \log^{\epsilon} n)$.
Therefore, the space complexity of the data structure is $O(n^2 \log^{\epsilon} n)$.

\subsubsection*{Top-$k$ Query Algorithm for Given $(P, [a, b], k)$:}
We first find the node $u=\locus(P)$ where the search for pattern $P$ in the suffix tree $\mathcal{T}$ terminates. Then, we find the $(k+1)$ points of $\mathsf{P}(u)$ in the rectangle $[a,b-|P|+1]\times (-\infty, +\infty)$ with smallest y-coordinates (i.e., smallest distances) using the range reporting data structure stored at node $u$.
Let $(i, j)$ be the consecutive occurrence corresponding to the (returned) point with the largest x-coordinate. If the copy of $P$ starting at index $j$ ends with a position greater than $b$, report the remaining $k$ consecutive occurrences as output.
Otherwise, we report the $k$ consecutive occurrences with smallest distances.
Finding $\locus(P)$ in the suffix tree $\mathcal{T}$ takes $O(|P|)$ time. If we are to find $k+1$ points with smallest y-coordinates in sorted order, the structure stored at a node $u$ takes $O(\log\log |\mathsf{P}(u)| + k + 1)$ time to answer a query. The additional step to scan the returned points and report $k$ consecutive occurrences takes $O(k+1)$ time. Thus, the total time needed to answer a top-$k$ query is $O(|P| + \log \log n + k)$.

\subsubsection*{Gap-bounded Query Algorithm for Given $(P, [a, b], [g_1, g_2])$:}
For gap-bounded query, the procedure is very similar.
Again, we first find the node $u = \locus(P)$ where the search for $P$ in the suffix tree $\mathcal{T}$ terminates.
Then, we query the data structure stored at node $u$ with the rectangle $[a,b-|P|+1]\times [g_1, g_2]$.
Let $S'$ be the set of the consecutive occurrences corresponding to the returned points.
Finally, we simply scan the set $S'$ and remove the consecutive occurrence $(i, j)\in S'$ with the largest first occurrence $i$
if $j > b-|P|+1$, in words,  the copy of $P$ at index $j$ does not fit entirely within the index range $[a, b]$.
The resulting set $S'$ is our output to the gap-bounded query.
The time spent to answer a gap-bounded query is $O(|P| + \log \log n + \out)$ where $\out$ denotes the number of outputs for the query. The analysis is very similar to that of top-$k$ query.
Thus, we have the following proposition.
\begin{proposition}\label{prop:naive}
  Given a given text $T$, one can build an $O(n^2\log^{\epsilon} n)$-space index that can be used to answer a top-$k$ query or a gap-bounded query in $O(|P| + \log\log n + \out)$ query time.
\end{proposition}

\section{Range \texorpdfstring{Top-$k$}{Top-k} Consecutive Occurrence Queries}\label{sec:topk}
In this section, we propose an algorithm that solves the top-$k$ query~(Definition~\ref{def:topk}) using subquadratic space.
Our approach is similar to those employed in \cite{NavarroT16,Bille22}.
The suffix tree $\mathcal{T}$ is decomposed using heavy path decomposition.
For each heavy path $h$ in the decomposition, we define a set of horizontal segments that compactly represents the sets $C(u)$ of consecutive occurrences associated with nodes $u\in h$.
We organize these segments into a data structure that supports top-$k$ queries for patterns $P$ with $\locus(P)$ on the path $h$.

\subsection{Solution}\label{sec:topk-DS}
Let $\mathcal{T}$ be the suffix tree built over the text $T$. The tree $\mathcal{T}$ is decomposed using the heavy path decomposition, which gives $n$ disjoint (heavy) paths. 
Let $h$ be a heavy path in the decomposition.
We denote by $\apex(h)$ the highest node in the path $h$. We say that a consecutive occurrence $(i, j)$ is \emph{present} at node $v$ if $(i,j)\in C(v)$. The following lemma is helpful in compactly representing the sets $C(v)$ of consecutive occurrences.
\begin{lemma}\label{obs:life-co}
  For each consecutive occurrence $(i, j)\in C(v), \text{ where } v\in h$, there exists two nodes $u$ and $w$ on the heavy path $h$ such that $(i, j)$ is present at each node of the subpath of $h$ joining $u$ and $w$ and is not present at any other node of $h$.
\end{lemma}
\begin{proof}
  Let $h =\langle \apex(h) = v_0, v_1, v_2,\ldots, v_t \rangle$, and suppose $(i, j)$ is present at node $v_d$ for some $d\in \{0, 1, \ldots, t-1\}$.  This implies that the leaves $\ell_i$ and $\ell_j$ are present in the subtree of $\mathcal{T}$ rooted at node $v_d$. As we move down from $v_d$ to $v_{d+1}$, one or both leaves, $\ell_i$ or $\ell_j$, might branch off from the heavy path $h$. If neither leaf $\ell_i$ nor $\ell_j$ is present in subtree rooted at $v_{d+1}$, the pair $(i, j)$ would not be part of the set $C(v_{d+1})$. Otherwise, $(i,j)$ would continue to exist at the next node $v_{d+1}$, i.e., $(i, j)\in C(v_{d+1})$. Note that if $(i,j)$ is not present at $v_{d+1}$, it will not appear in any descendant node of $v_{d+1}$ on the path $h$. This is because as we move down the path the existing leaves only disappear. See also Figure~\ref{fig:segment} for examples.
  \begin{figure}[tb]
    \centering
    \includegraphics[width=0.5\linewidth]{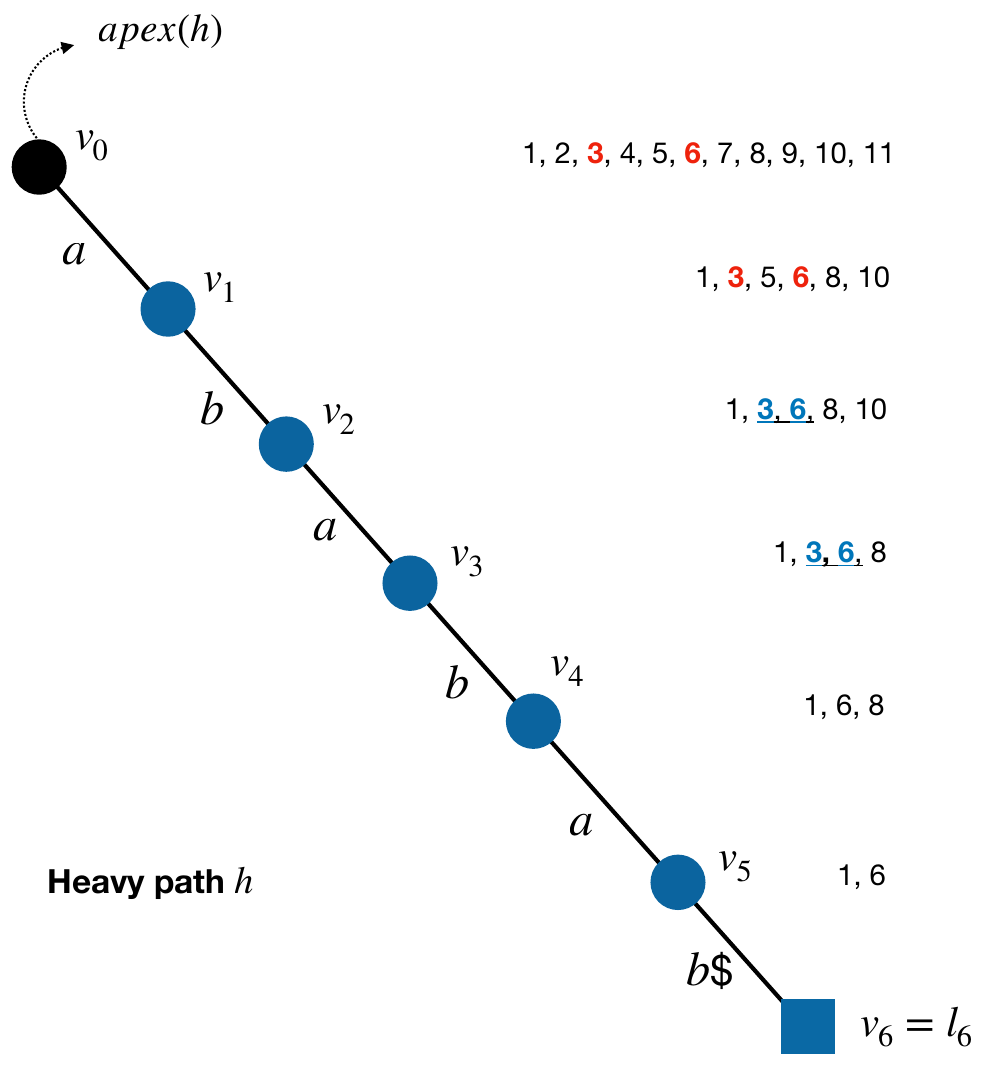}
    \caption{Illustration for Lemma~\ref{obs:life-co}. A heavy path $h$ in the suffix tree for $T = \mathtt{ababaababab\$}$ is shown in the figure. The occurrences $3$ and $6$ form a consecutive occurrence at node $v_2$; at each ancestor $v_0$ and $v_1$ they are separated by some indices. It remains a consecutive occurrence till node $v_3$ as leaf $l_3$ is no longer present in the subtree rooted at $v_4$. Thus, it is said that the consecutive occurrence $(3, 6)$ is alive on the subpath joining $v_2$ and $v_3$. The corresponding horizontal segment is $[2, 3] \times 3$; x-coordinates are determined by the depths of $v_2$ and $v_3$ and y-coordinate is determined by the distance of consecutive occurrence.}
    \label{fig:segment}
  \end{figure}

  Since $i$ and $j$ are two occurrences of $\str(v_d)$, the string $\str(v_0)$ must occur at these positions as $\str(v_0)$ is a prefix of $\str(v_d)$.
  If $\str(v_0)$ has no occurrence $g$ such that $i<g<j$, the pair $(i, j)$ is a consecutive occurrence of $\str(v_0)$, i.e., $(i, j)\in C(v_0)$. Otherwise, there are occurrences between $i$ and $j$, and the pair $(i, j)$ becomes a consecutive occurrence at the highest ancestor of $v_d$ where all occurrences between $i$ and $j$ disappear (i.e., the corresponding leaves branch off from the heavy path $h$).
\end{proof}
We denote by $\Pi_{i,j}^h$ the subpath of a heavy path $h$ corresponding to $(i, j)$, which is described in Lemma~\ref{obs:life-co}.
Let $u$ and $w$ be the starting and the ending nodes of path $\Pi_{i,j}^h$, respectively. 
Let $d(v)$ denote the depth of $v$ within heavy path $h$, i.e.,
$d(\apex(h)) = 0$ and $d(v) = d(v')+1$ where $v'$ is the parent of $v \ne \apex(h)$.
We create a horizontal segment $[d(u), d(w)]\times (j-i)$ in a 2D grid for each consecutive occurrence $(i, j)$ and associate $(i, j)$ as satellite data.
Let $I_h$ be the set of all such horizontal segments corresponding to path $h$.
By the construction of segments, the pair $(i, j)$ is a consecutive occurrence of pattern $P$ if and only if $\locus(P) \in h$ and $d(u) \le d(\locus(P))\le d(w)$ hold.
Also, since we are only interested in occurrences within substring $T[a.. b]$,
we will report such $(i, j)$ satisfying $a \le i < b-|P|+1$.
Thus, we can obtain the desired top-$k$ segments
by reporting $k+1$ segments\footnote{The reason for finding up to $k+1$ segments is the same as in Section~\ref{sec:naive}.} $[d(u), d(w)]\times (j-i)$ with the smallest y-coordinates
such that $d(u) \le d(\locus(P))\le d(w)$ and $a \le i < b-|P|+1$ hold.
Therefore, once $\locus(P), h$, and $d(\locus(P))$ in $h$ are computed,
the remaining process of a top-$k$ query can be reduced to the following geometric problem for $I = I_h$:
\begin{definition}[2D top-$k$ stabbing query with weight constraint]\label{def:geometrictopk} 
  A set $I = \{[l_1, r_1]\times y_1, \ldots, [l_{|I|}, r_{|I|}]\times y_{|I|}\}$ of horizontal segments in a 2D plane and weight function $w: I \to \mathbb{N}$ are given for preprocessing.
  The query is, given $(x, [w_1, w_2], k)$,
  to report the $k$ segments $s_1, \ldots, s_k$ with smallest y-coordinates in 
  non-decreasing order of y-coordinate
  such that $l_i \le x \le r_i$ and $w_1 \le w(s_i) \le w_2$ hold
  for every segment $s_i = [l_i, r_i] \times y_i$.
\end{definition}
We will later prove the following theorem in Section~\ref{sec:geometric}.
\begin{theorem} \label{thm:geometrictopk}
  There is a data structure of size $O(|I| \log |I|)$ that can answer the query of Definition~\ref{def:geometrictopk} in $O(\log |I| + k)$ time. Given set $I$, the data structure can be constructed in $O(|I| \log^2 |I|)$ time.
\end{theorem}

\subsubsection*{Space Complexity.}
We have the following result regarding the total size of the sets of segments.
\begin{lemma}[\cite{Bille22}]\label{lem:h-size}
  The total number of horizontal segments for all heavy paths is $O(n\log n)$.
\end{lemma}
By Theorem~\ref{thm:geometrictopk} and Lemma~\ref{lem:h-size}, the total size of our data structure is $O(n + \sum_{h} |I_h| \log |I_h|) = O(n \log^2 n)$.

\subsubsection*{Query Time.}
Since the locus of $P$ can be computed in $O(|P|)$ time using the suffix tree and $|I_h| \in O(n)$ holds for any heavy path $h$,
we obtain the following.
\begin{lemma}    
  We can represent a given text $T$ into an $O(n\log^{2}n)$ space data structure that supports a range top-$k$ consecutive occurrence query in $O(|P| + \log n + k)$ time.
\end{lemma}
Note that the query time of the lemma is optimal if $|P|\ge \log n$.
When $|P| < \log n$, the $\log$ factor dominates the pattern length $|P|$, and the query time becomes $O(\log n + k)$. 

Further, we slightly improve the query time using an additional data structure that efficiently answer queries with short patterns, i.e., $|P|< \log n$.
For each node $v$ with $|\str(v)| < \log n$ in the suffix tree $\mathcal{T}$, we store a data structure built over the set $C(v)$ of consecutive occurrences as in Section~\ref{sec:naive}.
The additional space used at each node $v$ is $O(|C(v)| \log ^{\epsilon}|C(v)|)$.
In the suffix tree $\mathcal{T}$, the total number of leaves in the subtrees rooted at the nodes of a particular level is $n$, so $\sum_{v: |\str(v)| < \log n}|C(v)| = O(n\log n)$.
Hence, the overall space used by the additional data structure is $O(n\log^{1+\epsilon} n)$, which is subsumed by $O(n\log^2 n)$.

For a query pattern $P$ with $|P|<\log n$, we simply find the node $\locus(P)$ in the suffix tree $\mathcal{T}$ in $O(|P|)$ time and query the associated range reporting data structure as in the na\"ive solution of Section~\ref{sec:naive}.
Thus, the query time would be $O(|P|+\log \log n + k)$.

\subsubsection*{Construction.}
Next, we present an algorithm that computes sets $I_h$ for all heavy paths $h$ in $O(n\log n)$ time and $O(n)$ space. By definition of horizontal segments in $I_h$, the $x$-range of the horizontal segment for a consecutive occurrence $(i, j)$ is determined by depths of the highest and deepest nodes of the path $\Pi_{i,j}^h$ in $h$ and its $y$-coordinate is equal to $j-i$. The construction algorithm considers each path $h$ in the decomposition one after another and computes the corresponding set $I_h$ of horizontal segments. The idea is to traverse the path $h$ in the top-down fashion and determine (highest and deepest nodes of) the path $\Pi_{i,j}^h$ for each consecutive occurrence $(i, j)$.

We now describe how to compute the set of horizontal segments for a particular heavy path. Let $h=\langle v_0, v_1,\ldots, v_t \rangle$ be a heavy path in the decomposition, where $v_0=apex(h)$ and $v_t$ is a leaf node in the suffix tree $\mathcal{T}$. We associate each node $v_p$ with the leaves that are present in the subtree $\mathcal{T}_p$ rooted at $v_p$, but not in $\mathcal{T}_{p+1}$. See Figure~\ref{fig:node-to-leaves}. The node $v_p$, for each $p=0,1,\ldots, t-1$, stores \emph{leaf pointers} to the leaves associated with it. The total number of leaf pointers is $O(n \log n)$ since any leaf to root path crosses at most $\log n$ heavy paths. Also, all the leaf pointers can be easily precomputed in $O(n \log n)$ in a bottom up manner.
\begin{figure}[tb]
  \centering
  \includegraphics[width=0.5\linewidth]{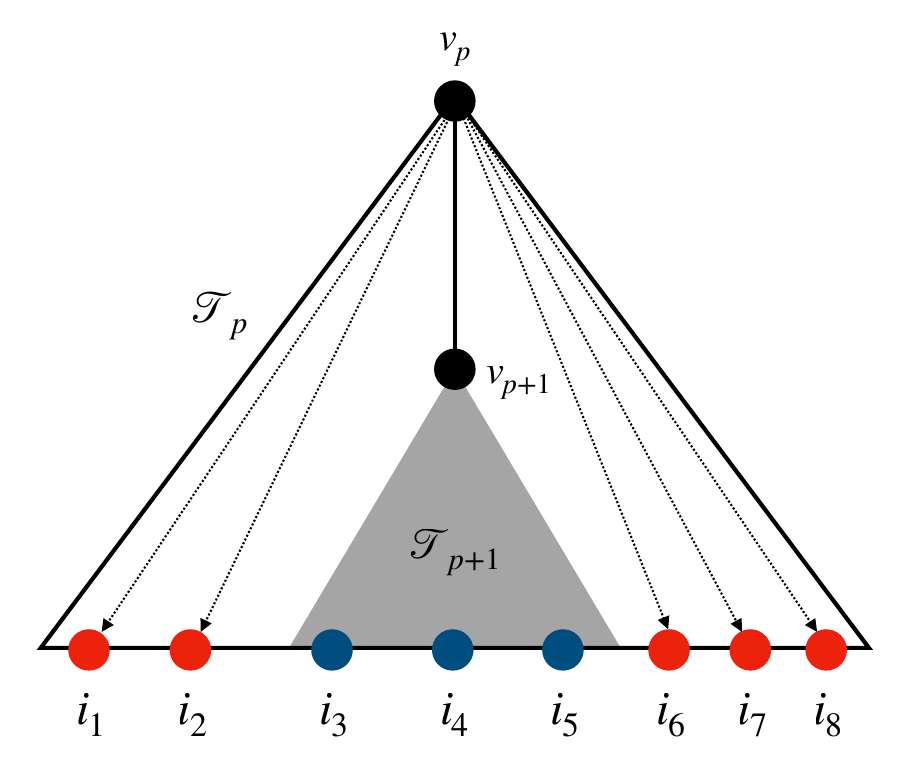}
  \caption{The subtree $\mathcal{T}_p$ has leaves $i_1, i_2, \ldots, i_8$ and $\mathcal{T}_{p+1}$ has leaves $i_3, i_4, i_5$, in blue color. Leaves $i_1, i_2, i_6, i_7$, and $ i_8$, in red color, are present in $\mathcal{T}_p$ but not in $\mathcal{T}_{p+1}$.
  The node $v_p$ stores leaf pointers to such leaves, shown as dotted arrows.}
  \label{fig:node-to-leaves}
\end{figure}
We next store the following information in a doubly linked list $\mathsf{L}_h$:
the labels of leaves in the subtree $\mathcal{T}_0$ rooted at $v_0 = \apex(h)$ sorted in increasing order.
Sorting such leaf labels for all $h$ can be performed in a total of $O(n \log n)$ time by applying radix sort to the $O(n \log n)$ pairs of nodes, where each pair consists of a leaf and the apex of a heavy path which is an ancestor of the leaf.
Note that using the leaf pointers, each node $v_p$ can identify (locate) its corresponding leaves in the list $\mathsf{L}_h$, in $O(1)$ time per leaf.
We now process the nodes $v_0, v_1,\ldots, v_t$ of path $h$, in order,  and determine, for each consecutive occurrence $(i,j)$, the highest and deepest nodes of the path $\Pi_{i,j}^h$ using the list $\mathsf{L}_h$. 
Initially, each consecutive pair of elements in $\mathsf{L}_h$ corresponds a consecutive occurrence of the string $\str(v_0)$, and $v_0$ serves as the highest node for each of them.
When we move from $v_p$ to $v_{p+1}$, we identify and delete the elements (leaves) from $\mathsf{L}_h$ that are associated with $v_p$ using the leaf pointers. For each deleted element $k$, at most two horizontal segments are included in the initially empty set $I_h$:
Suppose that $i$ and $j$ are the left and right neighbors of $k$ in the list $\mathsf{L}_h$, respectively.
Then, $(i, k)$ and $(k, j)$ are the consecutive occurrences at node $\str(v_p)$ to which $k$ contributes and they are no longer present at $v_{p+1}$.
In other words, $v_p$ is the deepest node at which $(i, k)$ and $(k, j)$ are present.
If $v_{p_1}$ and $v_{p_2}$ are the highest nodes on $h$ at which the consecutive occurrences $(i, k)$ and $(k, j)$ are present, respectively, we include the horizontal segments $[d(v_{p_1}), d(v_{p})]\times(k-i)$ and $[d(v_{p_2}), d(v_{p})]\times(j-k)$ into the set $I_h$. Here, $d(v_\ell )=\ell$ for $\ell=0,1,\ldots,t$.

Computing the leaf pointers to leaves from internal nodes and sorting labels of leaves can be done in $O(n \log n)$ time in total.
Let $n_h$ denotes the number of leaves in the subtree of $\mathcal{T}$ rooted at $\apex(h)$. The sorted list $\mathsf{L}_h$ occupies $O(n_h)$ space. Since each leaf is deleted from the list $\mathsf{L}_h$ only once and we are doing $O(1)$ amount of work for each deletion, so the overall time for all deletions would be $O(n_h)$. Therefore, the total running time for heavy path $h$ is $O(n_h)$. Summing over all heavy paths, we get the overall running time of the algorithm, which is $O(n\log n + \sum_h n_h) = O(n\log n + \sum_h |I_h|)= O(n\log n)$.
Thus, we have the following lemma.
\begin{lemma}\label{lem:preprocess-algo}
  Given a heavy path decomposition of the suffix tree $\mathcal{T}$, we can compute the set of horizontal segments for each heavy path in overall $O(n\log n)$ time.
  The proposed algorithm uses $O(n)$ working space.
\end{lemma}

The suffix tree $\mathcal{T}$ for the string $T$ can be constructed in $O(n\log n)$ time~\cite{Weiner73}, and it can be decomposed using the heavy path decomposition in $O(n)$ time. By Lemma~\ref{lem:preprocess-algo}, we can compute the set $I_h$ of horizontal segments for every heavy path $h$ in a total of $O(n\log n)$ time. 
Further, with each heavy path $h$ in the decomposition, the geometric structure of Theorem~\ref{thm:geometrictopk} over the set $I_h$ of horizontal segments is stored, which can be built in $O(|I_h| \log^2 |I_h|)$ time. Thus, the total time needed to build these geometric structures is $\sum_h O(|I_h| \log^2 |I_h|) = O(n \log^3 n)$, as $|I_h|\in O(n)$ and $\sum_h |I_h|=O(n\log n)$ (Lemma~\ref{lem:h-size}), which is the bottle neck. Thus, we obtain the main theorem of this paper:
\begin{theorem}\label{thm:topk}
  We can represent a given text $T$ into an $O(n\log^{2}n)$ space data structure that supports a range top-$k$ consecutive occurrence query in $O(|P| + \log\log n + k)$ time.
  The data structure can be constructed in $O(n \log^3 n)$ time.
\end{theorem}

\subsection{Proof of Theorem~\ref{thm:geometrictopk}: Solving Geometric Subproblem}\label{sec:geometric}

In this subsection, we prove Theorem~\ref{thm:geometrictopk}.
\subsubsection*{Data Structure.}
Given set $I$ of horizontal segments, we build a segment tree $\ST$ over $I$.
Each node $v\in \ST$ maintains the segments stored at node $v$, in a list $L_v$ sorted by their weights.
To efficiently perform predecessor search, we employ the fractional cascading technique (see also Section~\ref{sec:pre}) on the sorted lists $L_v$ for all $v \in \ST$.
We define an array $Y_v$ over the list $L_v$ such that $Y_v[i]$ stores the y-coordinate of segment $L_v[i]$.
We then preprocess the array $Y_v$ for the sorted range reporting queries of Lemma~\ref{lem:brodal}.
The size of the enhanced segment tree structure is $O(|I| \log |I|)$
since data structures associated with each node $v$ use $O(|L_v|)$ space
and $O(\sum_{v\in \ST}|L_v|) = O(|I|\log|I|)$ holds.

\subsubsection*{Query Algorithm.}
Given a query $(x, [w_1, w_2], k)$,
we first find the search path $\pi$ for $x$ in the segment tree $\ST$.
Let $\pi = \langle u_0 = \mathit{root}, u_1, \ldots, u_{|\pi|-1}\rangle$.
We then traverse the path $\pi$ in the top-down fashion and for each node $v \in \pi$, compute the positions (ranks) of $w_1$ and $w_2$ in list $L_v$,
respectively denoted by $\rnk_v(w_1)$ and $\rnk_v(w_2)$,
using the pointers provided by the fractional cascading technique described in Section~\ref{sec:pre}.
We prepare an incremental buffer $B$ of size $O(k + \log |I|)$ consisting of $|\pi| = O(\log |I|)$ rows.
Each $i$th row of $B$ corresponds to the node $u_i$ on the path $\pi$, denoted by $B[u_i]$.
Moreover, we maintain a min-heap $H$ implemented as a dynamic fusion node~(Lemma~\ref{lem:fusionnode}).
See Figure~\ref{fig:query} for illustration.
\begin{figure}[tbh]
  \centering
  \includegraphics[width=\linewidth]{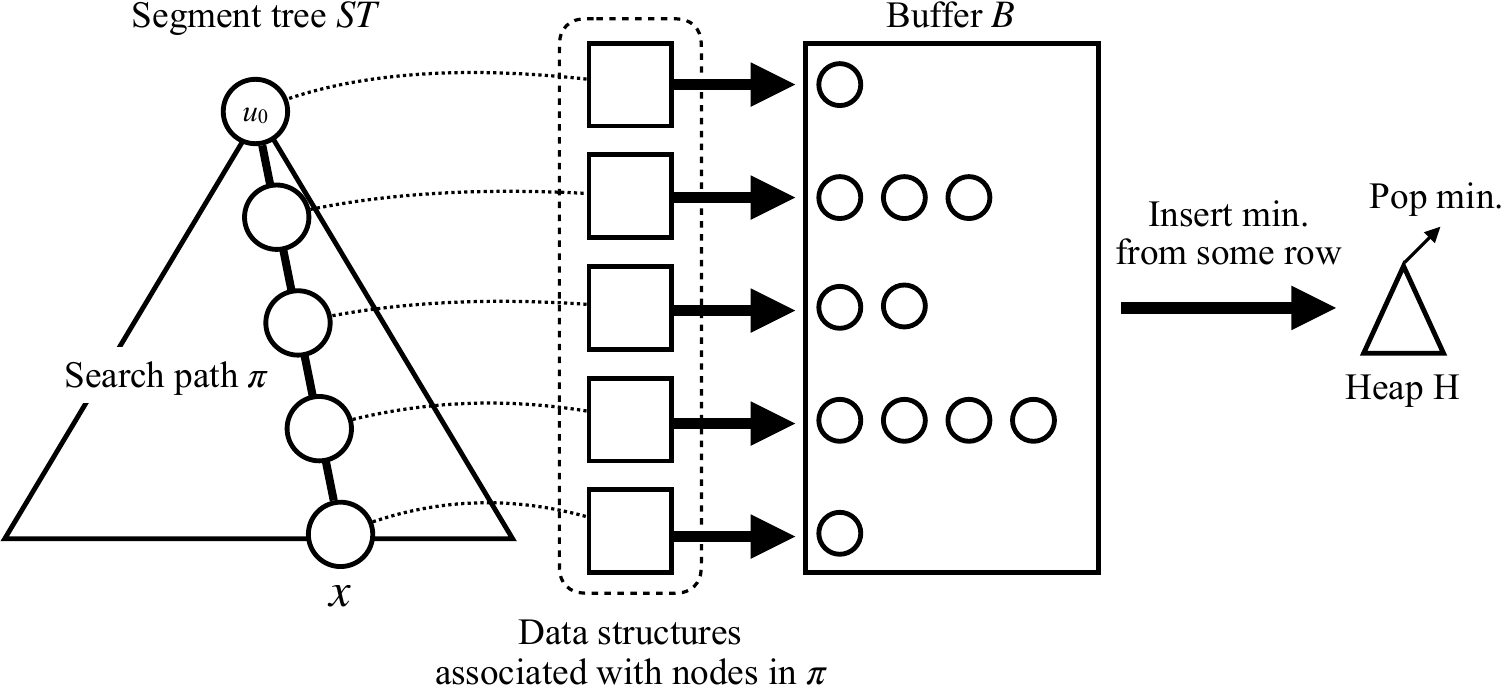}
  \caption{Illustration for an overview of the data structures used in the query algorithm.}
  \label{fig:query}
\end{figure}

We find top-$1$~(i.e., the smallest) items from the subarrays $Y_v[\rnk_v(w_1).. \rnk_v(w_2)]$ for all nodes $v \in \pi$ and
initialize the heap $H$ with the minima.
Also, we initialize the \emph{counter} $e(v) = 0$ for each node $v\in\pi$.
We then keep repeating the following steps until the heap $H$ gets empty or $k$ items are reported:
\begin{enumerate}
  \item Pop the minimum item $\mu$ from the heap $H$ and report it.
    Let $v$ be the node on the path $\pi$ from which the item $\mu$ came.
  \item If $v$'s buffer $B[v]$ is non-empty, simply pop the minimum item from $B[v]$ and insert it into the heap $H$. Then continue to the next iteration.
  \item Otherwise (if $B[v]$ is empty), find top-$2^{e(v)+1}$ items from $Y_v[\rnk_v(w_1).. \rnk_v(w_2)]$ by using~Lemma~\ref{lem:brodal}.
    Then remove the first half $2^{e(v)}$ of them, which have already been moved to $H$ in previous steps,
    and insert the remaining items to $B[v]$.
    Lastly, pop the minimum item from $B[v]$ and insert it into the heap $H$,
    and increment $v$'s counter $e(v)$ by $1$.
\end{enumerate}

\subsubsection*{Correctness.}
We focus on the changes to the buffer $B[v]$ for any fixed node $v \in \pi$ throughout the query algorithm.
Buffer $B[v]$ grows only if we enter Step~3, where $B[v]$ is empty.
Picking up only steps where $B[v]$ grows,
in each such step,
$B[v]$ is updated to the top-$2^{e(v)}$ items from the items in $Y_v[\rnk_v(w_1).. \rnk_v(w_2)]$
that are not inserted to $H$ before,
for incremental $e(v) = 0, 1, 2, \ldots$.
Also, when the smallest item from $B[v]$ is inserted to $H$ at Step~2, the item is removed from $B[v]$.
Hence, buffer $B[v]$ keeps the top-$c$ items from $Y_v[\rnk_v(w_1).. \rnk_v(w_2)]$ that have not been inserted to $H$ for some $c \ge 0$.
Also, the smallest \emph{unreported} item from $Y_v[\rnk_v(w_1).. \rnk_v(w_2)]$ is always present in heap $H$ for every $v \in \pi$ due to Step~2.
Thus, for each iteration of the three steps of the algorithm,
the smallest unreported item from $\pi$ that satisfies the weight constraint will be reported at Step~1 of the iteration.

\subsubsection*{Query Time.}
Next, we analyze the running time of the query algorithm.
Given $x$, determining the search path $\pi$ in the segment tree $\ST$ takes $O(\log |I|)$ time.
Also, thanks to the fractional cascading technique, all the ranks of $w_1$ and $w_2$ in the nodes in $\pi$ can be computed in a total of $O(\log |I|)$ time.
As the heap $H$ is implemented as dynamic fusion node structure, initializing $H$ takes $O(\log |I|)$ time.
While iterating the loop, when buffer $B[v]$ becomes empty (i.e., all items of $B[v]$ have either been moved to $H$ or reported as output) for a node $v$, it is replaced by a new buffer twice the size of the previous one in $O(|B[v]|)$ time by Lemma~\ref{lem:brodal}.
Since (1) the size of $B[v]$ is at most the number of items from $v$ already reported or moved to $H$ and
(2) the number of rows in $B$ is $|\pi| \in O(\log |I|)$,
the total size of the buffer $B$ is $O(k + \log |I|)$.
Further, each operation in a dynamic fusion node structure takes constant time.
To summarize, the total time spent\footnote{Note that if we implement the heap $H$ as a standard binary heap instead of the dynamic fusion node, the query time becomes $O((k + \log |I|)\log\log |I|)$ since the size of $H$ is $O(\log |I|)$.} to answer the query is $O(k + \log |I|)$.

\subsubsection*{Construction.}
We now analyze the construction of the enhanced segment tree for the set $I$.
Building the segment tree for the set $I$ takes $O(|I|\log |I|)$ time.
The fractional cascading data structure over the lists $L_v$ can be constructed in $O(|I|\log |I|)$ time and space (see Theorem~S in \cite{Chazelle86FC}).
By Lemma~\ref{lem:brodal}, the sorted range reporting data structure over the array $Y_v$ can be built using $O(|Y_v|\log |Y_v|)$ time and space.
Summing over all the nodes in the segment tree, we get the total time needed for this task, which would be equal to $\sum_v O(|Y_v|\log |Y_v|) = O(|I|\log^2 |I|)$.
This factor dominates the overall preprocessing time.

We have completed the proof of Theorem~\ref{thm:geometrictopk}.

\section{Range Gap-Bounded Consecutive Occurrence Queries}\label{sec:boundedgap}
In this section, we describe a space efficient data structure that supports gap-bounded queries~(Definition~\ref{def:gapbounded}).
The data structure is similar to the one described for top-$k$ queries in Section~\ref{sec:topk},
and it uses $O(n\log^{2+\epsilon}n)$ space. 
The query time is $O(|P|+\log\log n + \out)$.
The data structure differs from the top-$k$ structure described in Section~\ref{sec:topk} only in the representation of the sets $I_h$ of segments. 

We first build the suffix tree $\mathcal{T}$ for the text $T$ and decompose it using heavy path decomposition.
For each heavy path $h$ in the decomposition, we compute a set $I_h$ of horizontal segments that compactly represent the consecutive occurrences
as in Section~\ref{sec:topk-DS}.
Let $(P, [a, b], [g_1, g_2])$ be a query we should answer.
Firstly, we find the heavy path $h$ where $\locus(P)$ belongs.
Then, we can obtain the desired segments as follows:
\begin{enumerate}
  \item Compute segments $[d(u), d(w)]\times (j-i) \in I_h$ such that $d(u) \le d(\locus(P))\le d(w)$, $a \le i < b-|P|+1$, and $g_1 \le (j-i) \le g_2$ hold.
  \item Scan the consecutive occurrences corresponding to the obtained segments
    and remove (at most one) consecutive occurrence $(i ,j)$ with  $j > b-|P|+1$ if it exist~(as we mentioned in Section~\ref{sec:naive}).
  \item Report the remaining segments.
\end{enumerate}

The second and the third steps are obvious.
The first step can be reduced to the following problem for $I = I_h$,
which is different from that of the top-$k$ algorithm.
\begin{definition}[2D height-bounded stabbing query with weight constraint]\label{def:geometricgapbounded} 
  A set $I = \{[l_1, r_1]\times y_1, \ldots, [l_{|I|}, r_{|I|}]\times y_{|I|}\}$ of horizontal segments in a 2D plane
  and weight function $w: I \to \mathbb{N}$ are given for preprocessing.
  The query is, given $(x, [w_1, w_2], [\psi_1, \psi_2])$,
  to report all the segments $s_1, \ldots, s_t$ in 
  non-decreasing order of y-coordinate
  such that $l_i \le x \le r_i$, $\psi_1 \le y_i \le \psi_2$ and $w_1 \le w(s_i) \le w_2$ hold
  for every segment $s_i = [l_i, r_i] \times y_i$.
\end{definition}

We show the following theorem:
\begin{theorem} \label{thm:geometricgapbounded}
  There is a data structure of size $O(|I| \log^{1+\epsilon} |I|)$ that can answer the query
  of Definition~\ref{def:geometricgapbounded} in $O(\log |I| \log \log |I| + \out)$ time
  where $\epsilon > 0$ is a constant and $\out$ is the number of outputs.
  Given set $I$, the data structure can be constructed in 
  $O(|I|\log^{3/2} |I|)$ time.
\end{theorem}
\begin{proof}
  The idea is very similar to that of Theorem~\ref{thm:geometrictopk}.
  We construct a segment tree $\ST$ over $I$ and associate some data structures with nodes of the tree.
  Given a query $(x, [w_1, w_2], [\psi_1, \psi_2])$,
  we first find the search path $\pi$ for $x$.
  After that, the remaining operations related to $[w_1, w_2]$ and $[\psi_1, \psi_2]$ are considered on each node in $\pi$.
  The differences are:
  (1) we associate another data structure with nodes of the segment tree, and
  (2) we do not need to use a doubling technique with an incremental buffer in the query algorithm.
  Below, we focus only on the differences from Theorem~\ref{thm:geometrictopk}.

  Let $v$ be a node of the segment tree $\ST$.
  For each segment $s = [l, r]\times y$ in the set $S(v)$ of segments for $v$, we create a 2D point $p(s) = (w(s), y)$.
  We then construct a sorted orthogonal range reporting data structure of Lemma~\ref{lem:sortedrangereporting} on top of the set $\bigcup_{s\in S(v)} p(s)$ of points.
  The total size of the segment tree structure with range reporting structures is $O(|I|\log^{1+\epsilon} |I|)$.

  At each node $v$ of the segment tree $\ST$, construction of the sorted range reporting structure takes $O(|S(v)|\sqrt {\log |S(v)|})$ (Lemma~\ref{lem:sortedrangereporting}). So, the overall time for building this structure for all nodes in $\ST$ would take $\sum_{v\in \ST} O(|S(v)|\sqrt {\log |S(v)|})= O(|I|\log^{3/2} |I|)$, which subsumes $O(|I|\log |I|)$ time required to construct the segment tree $\ST$. Therefore, the construction time of the data structure is $O(|I|\log^{3/2} |I|)$.

  When a query is given, we traverse the search path $\pi$ for $x$.
  At each node $v$ in $\pi$,
  we create a (sorted) list $\sigma(v)$ of all the desired segments in $S(v)$
  by reporting all the points within the rectangle $[w_1, w_2] \times [\psi_1, \psi_2]$
  using the sorted orthogonal range reporting data structure.
  After finishing the traversal,
  we merge the sorted lists $\sigma(v)$ for all $v \in \pi$
  by using a heap implemented as dynamic fusion node as in Section~\ref{sec:topk},
  and report the segments in the merged list in order.
  The total time to answer a query is $O(\log |I| \log \log |I| + \out)$ time
  since the length of $\pi$ is $O(\log |I|)$ and range reporting query at each $v$ takes $O(\log \log |I| + |\sigma(v)|)$ time.

  The correctness can be seen from the definition of the 2D points $(w(s), y)$ within each node and the definition of the sorted range reporting query to be performed on those points.
\end{proof}

Now we obtain an $O(|P| + \log n\log\log n+ \out)$-query-time method for the range gap-bounded consecutive occurrence query using $O(n \log^{2+\epsilon} n)$ space since $\sum_h |I_h| \in O(n \log n)$ by Lemma~\ref{lem:h-size}.
We reduce the query time to $O(|P|+ \log\log n + \out)$ using the same technique employed for the top-$k$ case.
For each node $v$ in the suffix tree $\mathcal{T}$ with $|\str(v)|< \log n\log\log n$,
we build a data structure over the set $C(v)$ of consecutive occurrences as in Section~\ref{sec:naive}.
By the same arguments used in the top-$k$ case, the total space used by these additional data structures is
$O(n\log^{1+\epsilon}n\log\log n)$, which is dominated by $O(n\log^{2+\epsilon}n)$.

The suffix tree $\mathcal{T}$ over text $T$ can be built in $O(n\log n)$, and $\mathcal{T}$ can be decomposed into heavy paths in $O(n)$ time. 
The set $I_h$ of horizontal segments for each heavy path $h$ can be computed in a total of $O(n\log n)$ time using the algorithm presented in Lemma~\ref{lem:preprocess-algo}. 
The geometric data structure of Theorem~\ref{thm:geometricgapbounded} associated with each heavy path $h$ can be built in $O(|I_h|\log^{3/2} |I_h|)$ time. So, the total time for building geometric structures is $\sum_h O(|I_h|\log^{3/2} |I_h|) = O(n\log^{5/2} n)$, which dominates the other factors. Thus, the overall construction time of the final data structure is $O(n\log^{5/2} n)$.
Therefore, the next theorem holds;
\begin{theorem}\label{thm:gapbounded}
  We can represent a given text $T$ into an $O(n\log^{2+\epsilon}n)$ space data structure
  that supports a range gap-bounded consecutive occurrence query in $O(|P| + \log\log n + \out)$ time.
  The data structure can be constructed in $O(n\log^{5/2}n)$ time.
\end{theorem}

\section{Relation Between Consecutive Occurrences and Closed Substrings}\label{sec:closed_substrings}

In this section, we discuss a relation between consecutive occurrences of patterns and the notion of \emph{closed strings}, which is proposed by Fici~\cite{Fici17}.
A string is \emph{closed} if it is of length one or its longest border occurs exactly twice in the string. We observe that each consecutive occurrence of a pattern in a text corresponds to a closed substring of the text, and vice versa.
Let $(i, j)$ be a consecutive occurrence of pattern $P$ in the text $T[1..n]$, then $T[i..j+|P|-1]$ is a closed substring with border $T[i..i+|P|-1]$.
The maximum number of closed substrings in a length-$n$ string is $O(n^2)$ and this bound is tight~\cite{BadkobehFL15,ParshinaP24}.
Very recently, Jain and Mhaskar~\cite{Jain25} presented an $O(n \log n)$-time algorithm to compute a compact representation of all closed substrings of a string of length $n$ that uses only $O(n\log n)$ space,
which is based on the notion of \emph{maximal closed substrings}~\cite{Badkobeh0FP22}.
In this section, we give an alternative, indirect representation of the closed substrings, which also occupies $O(n\log n)$ space and propose an algorithm to compute it.

Lemma~\ref{obs:life-co} is about the presence of a consecutive occurrence on the nodes of a heavy path. More precisely, it states that in a heavy path, the nodes at which a consecutive occurrence is present forms of a subpath. The following lemma extends this result.
\begin{lemma}\label{lem:cons-occ-st}
  The nodes in the suffix tree $\mathcal{T}$ at which a consecutive occurrence $(i, j)$ is present constitutes an ancestor-descendant path.
\end{lemma}
\begin{proof}
  Let $\pi_i$ be the path from the root to leaf $i$.	Since all consecutive occurrences involving index $i$ are present only on the nodes of $\pi_i$, the node $w$ must lie on the path $\pi_i$. By the same argument used in the proof of Lemma~\ref{obs:life-co}, it follows that the nodes of path $\pi_i$ at which consecutive occurrence $(i, j)$ is present form a subpath of $\pi_i$. This completes the proof.
\end{proof}

We denote by $\Pi_{i,j}$ the maximal ancestor-descendant path in the suffix tree $\mathcal{T}$ that corresponds to a consecutive occurrence $(i, j)$. Let $u$ and $v$ be the starting and ending nodes of $\Pi_{i,j}$, respectively.
Note that for each substring $P$ which is a prefix of $\str(v)$ and has $\str(u)$ as a prefix, the pair $(i, j)$ is a consecutive occurrence of that substring. Equivalently, $T[i..j+|P|-1]$ is a closed substring of the underlying string $T$. So, the collection of all such paths compactly encode all the closed substrings of the string $T$. The number of all such ancestor-descendant paths would be $O(n\log n)$, as the number of consecutive occurrences is bounded above by $O(n\log n)$. 

We present an algorithm that efficiently computes all these paths $\Pi_{i,j}$. The algorithm uses the method for computing sets $I_h$ of horizontal segments as a subroutine. The idea is to compute the set $I_h$ for each heavy path $h$, and suitably combine the paths (corresponding to the horizontal segments) to obtain the paths $\Pi_{i,j}$. Thus, we get the following.
\begin{lemma}\label{lem:cons-occ-st-comp}
  Given the suffix tree $\mathcal{T}$ for a string, we can compute the path $\Pi_{i,j}$ for every consecutive occurrence $(i, j)$ present in the suffix tree $\mathcal{T}$ in $O(n\log n)$ time and space.
\end{lemma}
\begin{proof}
  We first compute the set $I_h$ for each heavy path in the decomposition, using our proposed algorithm.
  Let $\mathcal{I}=\bigcup_h I_h$. We then sort the segments in the set $\mathcal{I}$ by their associated consecutive occurrences using the radix sort algorithm. In the sorted list, the segments corresponding to a consecutive occurrence would be contiguous. On concatenating the subpaths (of heavy paths) corresponding to the segments of each consecutive occurrence $(i, j)$, we get the path $\Pi_{i,j}$ for the pair $(i, j)$ in the suffix tree $\mathcal{T}$.

  The algorithm for computing the set $\mathcal{I}$ of all horizontal segments takes $O(n\log n)$ time by Lemma~\ref{lem:preprocess-algo}.
  The radix sort for the set $\mathcal{I}$ would take $O(n + |\mathcal{I}|)= O(n\log n)$ time. The last step of concatenating subpaths for every consecutive occurrence can be done in overall $O(|\mathcal{I}|) = O(n\log n)$ time.
  Therefore, the total time is $O(n\log n)$.
\end{proof}

To summarize, we obtain the following:
\begin{proposition}\label{prop:closed}
  The set of closed substrings of a string $T$ of length $n$
  can be implicitly represented as $O(n \log n)$ paths
  over the suffix tree of $T$.
  Given string $T$, we can compute the representation in $O(n \log n)$ time.
\end{proposition}

Note that the asymptotic complexities in Proposition~\ref{prop:closed} are the same as those of the methods proposed in~\cite{Badkobeh0FP22} and~\cite{Jain25}.
However, our approach is based on textbook algorithms and admits a straightforward complexity analysis, whereas both existing methods require advanced dynamic data structures~\cite{BrownT79,KociumakaPRRW15}.

\section{Conclusions and Future Work} \label{sec:conclusions}
In this paper, we introduced range variants of consecutive occurrence queries,
generalizing the problems studied in previous work~\cite{NavarroT16,Bille22}.
Specifically, we proposed:
\begin{itemize}
  \item A data structure of size $O(n \log^2 n)$ that answers the range top-$k$ consecutive occurrence query in $O(|P| + \log\log n + k)$ time.
    Its construction time is $O(n\log^3 n)$.
  \item A data structure of size $O(n \log^{2+\epsilon} n)$ that answers the range gap-bounded consecutive occurrence query in $O(|P| + \log\log n + \out)$ time, where $\epsilon$ is a positive constant and $\out$ denotes the number of outputs.
    The structure can be built in $O(n\log^{5/2} n)$ time.
\end{itemize}
Further, we presented efficient algorithms for geometric problems involving weighted horizontal segments in a 2D plane, which appeared as subproblems in our approach and are of independent interest.
Additionally, we noticed that the notions of consecutive occurrences and closed substrings are related. Using this relation, we provided an indirect compact representation of all closed substrings in a string, which, in a sense, is equivalent to the one recently proposed by Jain and Mhaskar~\cite{Jain25}. We also
presented an $O(n\log n)$ time algorithm to compute the representation.

Our data structures require $\Omega(n \log^2 n)$ space in the worst case, whereas existing solutions for non-range variants use only $O(n \log n)$ space.
Thus, reducing the space complexity of our data structures is an interesting direction for future research.
Another promising direction for future work is to explore alternative time-space trade-offs for the problems under consideration, with a focus on achieving optimal query time and/or (near) linear space.

\end{document}